\documentclass[10pt,conference]{IEEEtran}

\usepackage{graphicx}
\usepackage{bm}
\usepackage{enumerate}
\usepackage{float}
\usepackage{multicol}
\usepackage{color}
\usepackage{url}
\usepackage{hyperref}
\usepackage{cite}
\pdfoutput=1

\usepackage{amsmath}
\usepackage{amsthm}
\usepackage{amssymb}

\usepackage{algorithmic}
\usepackage{algorithm}
\usepackage{acronym}

\newtheorem{lemma}{Lemma}
\usepackage[english]{babel}
\usepackage{array}

\ifCLASSOPTIONcompsoc
  \usepackage[caption=false,font=normalsize,labelfont=sf,textfont=sf]{subfig}
\else
  \usepackage[caption=false,font=footnotesize]{subfig}
\fi

\acrodef{4G}{4th Generation}
\acrodef{5G}{5th Generation}
\acrodef{6G}{6th Generation}
\acrodef{BS}{base station}
\acrodef{RA}{receive antenna}
\acrodef{PA}{predictor antenna}
\acrodef{IID}{independent and identically distributed}
\acrodef{CDF}{cumulative distribution function}
\acrodef{PDF}{probability density function}
\acrodef{cu}{channel use}
\acrodef{ACK}{acknowledgement}
\acrodef{NACK}{negative acknowledgement}
\acrodef{SNR}{signal-to-noise ratio}
\acrodef{HARQ}{hybrid automatic repeat request}
\acrodef{CSIT}{channel state information at the transmitter side}
\acrodef{CSI}{channel state information}
\acrodef{INR}{incremental redundancy}
\acrodef{npcu}{nats-per-channel-use}
\acrodef{MIMO}{multiple input multiple output }
\acrodef{TDD}{time division duplex}
\acrodef{E2E}{end-to-end}
\acrodef{IAB}{integrated access and backhaul}
\acrodef{ICNIRP}{International Commission on non-ionizing radiation protection}
\acrodef{3GPP}{3rd Generation Partnership Project}
\acrodef{MR}{moving relay}
\acrodef{QoS}{quality of service}
\acrodef{VPL}{vehicle penetration loss}
\acrodef{LoS}{line-of-sight}
\acrodef{NLoS}{non-line-of-sight}
\acrodef{mmWave}{millimeter wave}
\acrodef{CoMP}{Coordinated Multi-Point}
\acrodef{MISO}{multiple input single output }
\acrodef{BF}{beamforming}
\acrodef{FBL}{finite block-length}
\acrodef{UE}{user equipment}
\acrodef{UL}{uplink}
\acrodef{DL}{downlink}
\acrodef{RIS}{reconfigurable intelligent surface}
\acrodef{IU}{intended user}
\acrodef{IUE}{intended user equipment}
\acrodef{NIU}{non-intended user}
\acrodef{NIUE}{non-intended user equipment}
\acrodef{EMF}{electromagnetic field}
\acrodef{EMFE}{electromagnetic field exposure}
\acrodef{DFT}{discrete Fourier transform}
\acrodef{MRT}{maximum-ratio transmission}
\acrodef{AO}{alternating optimization}
\acrodef{AoD}{angle-of-departure}
\acrodef{JBPO}{joint beamforming and transmit power optimization}

\begin{document}
\captionsetup{belowskip=0pt,aboveskip=0pt}

\title{Electromagnetic Field Exposure Avoidance thanks to Non-Intended User Equipment and RIS}

\author{\IEEEauthorblockN{Hao~Guo\IEEEauthorrefmark{1}, Dinh-Thuy Phan-Huy\IEEEauthorrefmark{2}, and
Tommy~Svensson\IEEEauthorrefmark{1}}
\IEEEauthorblockA{\IEEEauthorrefmark{1}Chalmers University of Technology, 41296 Gothenburg, Sweden, \{hao.guo, tommy.svensson\}@chalmers.se}% <-this % stops an unwanted space
\IEEEauthorblockA{\IEEEauthorrefmark{2}Orange Innovation/Networks, 92326 Chatillon, France, dinhthuy.phanhuy@orange.com}
}

\maketitle
%--------------------------------------------------
\begin{abstract}
On the one hand, there is a growing demand for high throughput which can be satisfied thanks to the deployment of new networks using massive multiple-input multiple-output (MIMO) and beamforming. On the other hand, in some countries or cities, there is a demand for arbitrarily low electromagnetic field exposure (EMFE) of people not concerned by the ongoing communication, which slows down the deployment of new networks. Recently, it has been proposed to take the opportunity, when designing the future 6th generation (6G), to offer, in addition to higher throughput, a new type of service: arbitrarily low EMFE. Recent works have shown that a reconfigurable intelligent surface (RIS), jointly optimized with the base station (BS) beamforming can improve the received throughput at the desired location whilst reducing EMFE everywhere. In this paper, we introduce a new concept of a non-intended user (NIU). An NIU is a user of the network who requests low EMFE when he/she is not downloading/uploading data. An NIU lets his/her device, called NIU equipment (NIUE), exchange some control signaling and pilots with the network, to help the network avoid exposing NIU to waves that are transporting data for another user of the network: the intended user (IU), whose device is called IU equipment (IUE). Specifically, we propose several new schemes  to maximize the IU throughput under an EMFE constraint at the NIU (in practice, an interference constraint at the NIUE). Several propagation scenarios are investigated. Analytical and numerical results show that proper power allocation and beam optimization can remarkably boost the EMFE-constrained system's performance with limited complexity and channel information.
\end{abstract}
\begin{IEEEkeywords}
Reconfigurable intelligent surface (RIS), massive multiple-input-multiple-output (M-MIMO), electromagnetic field exposure (EMFE), beamforming, millimeter wave (mmWave).
\end{IEEEkeywords}

%--------------------------------------------------
\section{Introduction} % 1.5 Page with abstract
The \ac{5G} and most likely coming \ac{6G} use massive \ac{MIMO} beamforming and \ac{mmWave} \cite{vook2014ims} to boost the spectral efficiency and the energy efficiency of networks. However, in some countries or cities, due to a local and stronger concern regarding human exposure to \ac{EMF}, the \ac{EMFE} limits set by the local regulation can reach a very low level, sometimes ten times lower \cite[Fig. 10]{luca2021health} than the levels recommended by \ac{ICNIRP} \cite{icnirp2020limit}. Such arbitrarily low levels slow down the deployment of new networks, as it has been observed as early as with \ac{4G} \cite{gsma2014arbitraty}. For these reasons, in \cite{emilio2021reconfigurable,george2022smart}, it has been proposed to take the opportunity when designing the  new generation of networks, the future \ac{6G}, to offer a new type of service: arbitrarily low \ac{EMFE} as a service. 

To do so, \cite{emilio2021reconfigurable,george2022smart} proposes to use \acp{RIS}, to shape the radio propagation environment \cite{renzo2019smart,basar2019wireless}. \ac{RIS} are new network nodes reflecting waves in the desired direction and can be seen as intelligent mirrors or passive relays \cite{renzo2020reconfigurable}. In \cite{luca2021health,hussam2022emf,zappone2022spl,dt2022spawc}, \ac{RIS}-aided schemes have been proposed to reduce the self-\ac{EMFE} of a customer to his/her own \ac{UE} transmissions in the \ac{UL} direction. In \cite{awarkeh2021spawc,nour2022eucnc,nour2022eucncusing,yu2022swirnet}, \ac{RIS}-aided \ac{DL} beamforming schemes have been proposed to confine the over-exposed area, due to the transmission of the \ac{BS} in the \ac{DL} direction, inside a predefined circle. To our best knowledge, up to now, no study has been performed about the full potential of an \ac{EMFE}-constrained \ac{RIS}-assisted system. 

To our best knowledge, no study has been performed so far on an \ac{EMFE}-constrained \ac{RIS}-assisted system where the exposed person helps the network control his/her \ac{EMFE}. In this paper, for the first time, we introduce the concept of \ac{NIU} who is a customer of a mobile network operator with no data needs at the considered moment. The device carried by a \ac{NIU} is called \ac{NIUE}. Without data transmission, the \ac{NIUE} agrees to let the network use some control signaling and pilot exchanges between itself and the network while limiting the \ac{EMFE} at the \ac{NIU}, by controlling the level of interference at the \ac{NIUE}. In this work, we study how the \ac{NIUE} and a \ac{RIS} could help the \ac{BS} maximize the throughput at the \ac{IUE} whilst meeting an \ac{EMFE} constraint at the \ac{NIU} (or equivalently meeting an interference constrain at the \ac{NIUE}).

In order to realize this purpose, we propose different methods to maximize the throughput of an \ac{IU} and at the same time meet an \ac{EMFE} requirement (i.e. a threshold) at an \ac{NIU} side (this \ac{EMFE} being caused by the \ac{DL} transmission between a \ac{BS} and an \ac{IU}), thanks to the help of the \ac{NIUE}, the \ac{UE}, and an \ac{RIS}. Our contributions are as follows:
\begin{itemize}
    \item We first formulate a \ac{JBPO} problem at the \ac{BS} side, with the objective to provide the best data rate at the \ac{IU} under the \ac{NIU} \ac{EMFE} constraint. 
    \item We then propose two schemes to solve the \ac{JBPO} problem with different performance and different assumptions on \ac{CSI} availability: \ac{AO}-, and \ac{DFT}-based optimization. They are later combined with five power allocation methods between the direct path and the \ac{RIS}-assisted path. Finally, we provide analytical insights on 1) performance upper bound with transmitting different codewords as well as 2) a simple power-filling rule with priority on the direct-link transmission.
    \item Two propagation scenarios are investigated: 1) the \ac{NIUE} is located on the direct propagation path between the \ac{BS} and the \ac{IU}, and 2) the \ac{NIUE} is located out of this path alongside the \ac{IU}. 
\end{itemize}

The remainder of the paper is organized as follows. Section II introduces the system model and the problem formulation of \ac{JBPO}. In Section II, the \ac{AO}- and \ac{DFT}-based optimization schemes are proposed and analyzed. Section IV presents simulation results to validate the performance of the proposed methods while Section V concludes the paper.

%--------------------------------------------------
\section{System Model} % 1 Page
Consider an \ac{RIS}-assisted \ac{MIMO} \ac{DL} system with one \ac{BS} deployed with $N_\text{T}$ transmit antennas, one \ac{RIS} with $N$ elements, and a pair of single-antenna \ac{IU} and \ac{NIU}. Here, \ac{IU} is the user who wants to receive service from the \ac{BS} while the \ac{NIU} does not want to connect to the \ac{BS}. More importantly, \ac{NIU} prefers to have limited \ac{EMFE}  given that it requires no communication service. We assume that the \ac{NIU} does not decode the message or absorb the energy from the \ac{BS} and acts as one probe in space to detect the signal strength. Different from recent studies, e.g., \cite{awarkeh2021spawc,zappone2022spl},  as shown in Figs. \ref{fig_1}-\ref{fig_2}, two cases are considered in this work:
\begin{itemize}
    \item \textbf{Case 1}: The \ac{NIU} moves between the \ac{BS} and the \ac{IU}. Depending on the location of the \ac{NIU} and the beam width from the \ac{BS}, at some points/areas, the \ac{NIU} and \ac{IU} could receive the same transmission beam from the \ac{BS}.
    \item \textbf{Case 2}: The \ac{NIU} moves alongside the \ac{IU} and it could receive a different beam/side beam from the \ac{BS}.
\end{itemize}

\begin{figure}
\centering
  \includegraphics[width=1.0\columnwidth]{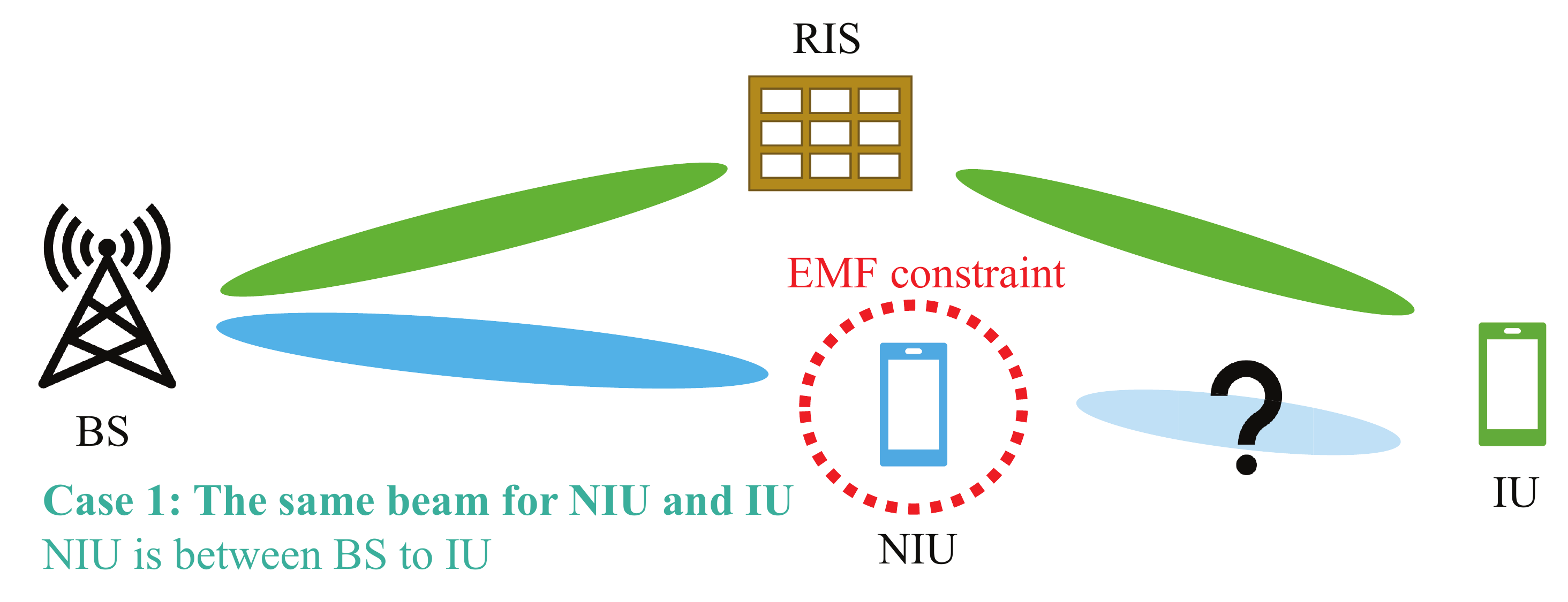}\\
\caption{\ac{EMFE}-constrained \ac{RIS} system, Case 1: The \ac{NIU} moves between the \ac{BS} and \ac{IU} and it receives the same beam as the \ac{IU}.}
\label{fig_1}
\end{figure}

In this way, the received message for the \ac{IU} is
\begin{align}
    y = \sqrt{P} \bm{h}\bm{w}x + n,
\end{align}
where $P$ is the transmit power at the \ac{BS} and $\bm{w}x$ is  the transmitted message with a unit-power precoder.  $n$ represents the additive Gaussian noise at the receiver side.  The channel $\bm{h}\in \mathcal{C}^{1\times N_\text{T}} $ between \ac{BS}-\ac{IU} includes both the direct link (\ac{BS}-\ac{IU}) as well as the \ac{RIS}-assisted link (\ac{BS}-\ac{RIS}-\ac{IU}), i.e.,
\begin{align}
    \bm{h} = \bm{h}_{\text{D}} + \bm{h}_{\text{RI}}\bm{\Theta}\bm{h}_{\text{BR}}.
\end{align}
Here, $\bm{h}_{\text{D}}\in\mathcal{C}^{1\times N_\text{T}}$ is the direct link, while $\bm{h}_{\text{BR}}\in\mathcal{C}^{N\times N_\text{T}}$ and $\bm{h}_{\text{RI}}\in\mathcal{C}^{1\times N}$ are the channel between \ac{BS}-\ac{RIS} and \ac{RIS}-\ac{IU}, respectively. We assume that $\bm{h}_{\text{D}}$ is independent of $\bm{h}_{\text{BS}}$ and $\bm{h}_{\text{RI}}$. Moreover,
\begin{align}
    \bm{\Theta}= \text{diag}(e^{j\phi_1}, ..., e^{j\phi_{N}})
\end{align}
is the reflection coefficient matrix of the \ac{RIS}. Assuming the channels of \ac{BS}-\ac{IU} and \ac{BS}-\ac{RIS}-\ac{IU} are spatially well separated/orthogonal such that the cross terms can be omitted, with \ac{BS}-active and \ac{RIS}-passive beamforming, the data rate at the \ac{IU} can be expressed as \footnote{Note that, by ignoring the cross terms we are actually presenting the lower bound of the system performance, whereas the exact performance with joint precoder design is left for further work.}
\begin{align}\label{eq_R_defi}
    R = B\log_2\left(  1 +  \frac{P_\text{D}\left|\bm{h}_{\text{D}}\bm{w}_{\text{D}}\right|^2 + P_\text{R}\left|\bm{h}_{\text{RI}}\bm{\Theta}\bm{h}_{\text{BR}}\bm{w}_{\text{R}}\right|^2 }{BN_0}      \right),
\end{align}
where $P_\text{D}$, $P_\text{R}$, $\bm{w}_{\text{D}}\in\mathcal{C}^{N_\text{T}\times 1}$, $\bm{w}_{\text{R}}\in\mathcal{C}^{N_\text{T}\times 1}$  are the power of the direct link, the power of the \ac{RIS} link, the precoder of the direct link, and the precoder of the \ac{RIS} link, respectively. Both $\bm{w}_{\text{D}}$ and $\bm{w}_{\text{R}}$ have unit power. Also, $B$ and $N_0$ represent the system bandwidth and the noise power spectral density, respectively.

\begin{figure}
\centering
  \includegraphics[width=1.0\columnwidth]{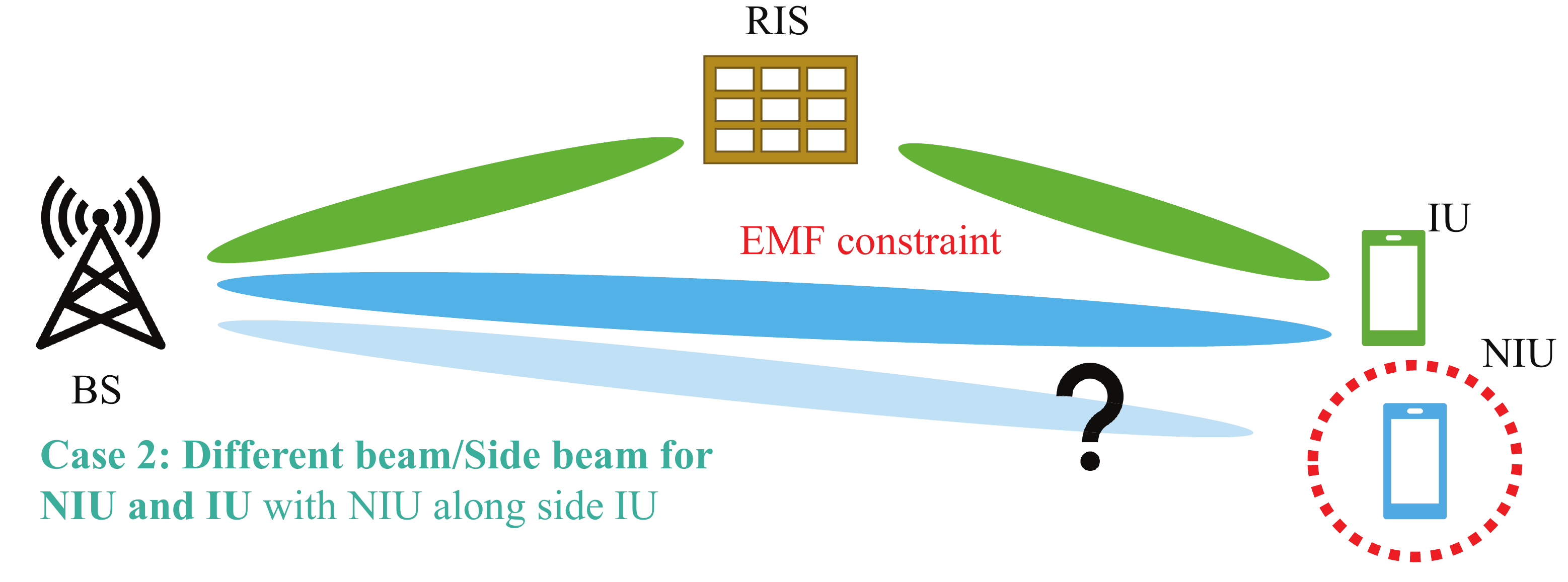}\\
\caption{\ac{EMFE}-constrained \ac{RIS} system, Case 2: The \ac{NIU} is alongside the \ac{IU} and it receives a different beam/side beam compared to the \ac{IU}. }
\label{fig_2}
\end{figure}

Assuming the \ac{NIU} occupies a unit area, the \ac{EMFE} constraint in terms of received power can be written as \cite{luca2022twc}
\begin{align}\label{eq_power_niu}
    P_{\text{N, Rx}} = \frac{4\pi}{\lambda}P_{\text{N, Tx}}\left|\bm{h}_\text{N}\bm{w}_\text{N}\right|^2,
\end{align}
with $P_{\text{N, Tx}}$ ($P_{\text{N, Rx}}$) being the transmit (receive) power of the \ac{BS}-\ac{NIU} link, and $\bm{h}_{\text{N}}$ and $\bm{w}_{\text{N}}$ are the channel and the precoder of the \ac{BS}-\ac{NIU} link, respectively. Both direct and reflected paths  can be considered using such a model. We assume that the channel of the \ac{BS}-\ac{NIU} link is known at the \ac{BS}, and the \ac{NIUE} is capable of sensing the received power and reporting back to the \ac{BS}.

%--------------------------------------------------
\section{EMFE-aware Joint Beamforming and Transmit Power Optimization} % 1 Page
The target of this work is to jointly maximize the rate at the \ac{IU} with proper beamforming design ($\bm{w}_{\text{D}}$, $\bm{w}_{\text{R}}$, $\bm{\Theta}$) and power allocation ($P_\text{D}$, $P_\text{R}$), in the presence of the \ac{NIU} with an \ac{EMFE} constraint. It can be mathematically stated as
\begin{equation}\label{eq_R}
\begin{aligned}
\max_{P_\text{D}, P_\text{R}, \bm{w}_\text{D}, \bm{w}_\text{R}, \bm{\Theta}} \quad & R\\
\textrm{s.t.} \quad & P_{\text{N, Rx}}\leq \bar{P}, \\
&0\leq P_\text{D}\leq P, 0\leq P_\text{R}\leq P,
   P_\text{D} + P_\text{R} = P,\\
  & \|\bm{w}_\text{D}\|^2 = 1, \|\bm{w}_\text{R}\|^2 = 1,\\
  & \phi_n \in [0,2\pi], n = 1, ..., N, \forall{n}.
\end{aligned}
\end{equation}
Here, $\bar{P}$ is a constant which represents the \ac{EMFE} constraint in terms of the received power at the \ac{NIU} and can be modeled by, e.g., (\ref{eq_power_niu}) as a function of $\bm{h}_{\text{N}}$. Since the beamforming optimization is independent of the transmit power, in the following, we first jointly optimize the beam of the \ac{BS} and the \ac{RIS} using two schemes with different \ac{CSI} requirements. Then, we propose different methods to allocate the transmit power at the \ac{BS} considering the \ac{EMFE} requirement of \ac{NIU}.

\subsection{Joint Beam Optimization for the BS and RIS}
\addtolength{\topmargin}{0.02in}
With the assumption of independent paths of the direct link and the \ac{RIS}-assisted link, the beamforming pattern of the direct link can be obtained by, e.g., \ac{MRT}, as $\bm{w}_\text{D} = \bm{h}_\text{D}^{H}/\left\|\bm{h}_\text{D}\right\|$, with $(\cdot)^{H}$ representing the Hermitian transpose.

For the \ac{RIS}-aided indirect link, we need to jointly optimize the active beamforming vector $ \bm{w}_\text{R}$ at the \ac{BS}, as well as the passive beamforming matrix $\bm{\Theta}$ at the \ac{RIS}. This problem has been widely studied in the literature under different setups/assumptions. For instance, with perfect knowledge of \ac{CSI}, i.e., $\bm{h}_{\text{BR}}$ and $\bm{h}_{\text{RI}}$, \ac{AO} has been shown to converge to the optimal solution of (\ref{eq_R}) \cite[Sec. III-B]{wu2019twc} \cite[Sec. III]{guo2020twc}\cite[Algorithm I]{zappone2021twc}. In this work, we use \ac{AO} in \ac{EMFE}-aware \ac{RIS} systems as an upper bound to compare to, as presented in Algorithm. \ref{alg_1}. Here, maximizing $R$ is equivalent to optimize $|\bm{h}_{\text{RI}}\bm{\Theta}\bm{h}_{\text{BR}}\bm{w}_{\text{R}}|$. Using \ac{AO}, for a fixed $\bm{w}_{\text{R}}$, the optimal $\phi_n$ resulting in $-\angle \left\{ \bm{h}_{\text{RI}}^{*}(n)\bm{h}_\text{w}(n), \forall{n} \right\}$ with $\bm{h}_\text{w} = \bm{h}_{\text{BR}}\bm{w}_\text{R}$. Then, for fixed $\Theta$, the optimal $\bm{w}_{\text{R}}$ can be obtained by calculating the dominant right eigenvector of $\bm{g} = \bm{h}_{\text{RI}}\bm{\Theta}\bm{h}_\text{BR}$. Once the optimal beams of the \ac{BS} and \ac{RIS} are settled, various power allocation methods can be applied to reach the maximum rate (\ref{eq_R_defi}), which are presented in Sec. \ref{sec_power}.

\begin{algorithm}[t!]
 \caption{Joint beam optimization using \ac{AO} in \ac{EMFE}-aware \ac{RIS} systems}
 \begin{algorithmic}
  \REQUIRE $\bm{h}_{\text{BR}}$, $\bm{h}_{\text{RI}}$, $\bm{h}_{\text{N}}$, and $\bar{P}$\\
 1. Initialize $ \bm{w}_\text{R}$ to some feasible values at the \ac{BS}.
\REPEAT
\STATE2. {Calculate $\bm{h}_\text{w} = \bm{h}_{\text{BR}}\bm{w}_\text{R}$};\\
\STATE3. {Set the \ac{RIS} phase $\phi_n = -\angle \left\{ \bm{h}_{\text{RI}}^{*}(n)\bm{h}_\text{w}(n), \forall{n} \right\}$};\\
\STATE4. {Compute $\bm{g} = \bm{h}_{\text{RI}}\bm{\Theta}\bm{h}_\text{BR}$};\\
\STATE5. {Set $\bm{w}_\text{R}$ as the right dominant eigenvector of $\bm{g}$};\\
\UNTIL{Convergence}\\
6. Compute power allocation $P_{\text{D}}$ and $P_{\text{R}}$ use methods in Sec. \ref{sec_power}, considering \ac{EMFE} constraint $\bar{P}$ under channel $\bm{h}_{\text{N}}$.
\RETURN Optimal rate $R$ for \ac{IU} (\ref{eq_R_defi}).
 \end{algorithmic}
 \label{alg_1}
 \end{algorithm}
 
In some scenarios, especially with large number of \ac{RIS} elements, optimizing $\phi_n$ and $\bm{w}_\text{R}$ with explicit \ac{CSI} may not be practical. Inspired by the precoding scheme designed in \cite[Algorithm 1]{huang2021globecom}, we propose to use a \ac{DFT} codebook-based beam optimization where the \ac{RIS} beam is selected from the pre-defined beam patterns while only the concatenated channel $\bm{h}_{\text{RI}}\bm{\Theta}\bm{h}_{\text{BR}}\bm{w}_{\text{R}}$ is needed to optimize $\bm{w}_{\text{R}}$, as presented in Algorithm \ref{alg_2}.

Complexity: The computational complexity of Algorithm. \ref{alg_1} is $\mathcal{O}\left(I_{\text{AO}}\left(N^2+2NN_{\text{T}}+N+N_{\text{T}}\right)\right)$, where $I_{\text{AO}}$ represents the number of iterations. For Algorithm. \ref{alg_2} the complexity is $\mathcal{O}\left(\left(N+1\right)\left(N^2+NN_{\text{T}}\right)\right)$.

\begin{algorithm}[t!]
 \caption{Joint beam optimization using \ac{DFT} codebook-based beamforming in \ac{EMFE}-aware \ac{RIS} systems}
 \begin{algorithmic}
 \REQUIRE Concatenate channel $\bm{h}_{\text{RI}}\bm{\Theta}\bm{h}_{\text{BR}}\bm{w}_{\text{R}}$ with selected $\bm{\Theta}$, pre-defined \ac{DFT} codebook $\bm{V}\in\mathcal{C}^{N\times N}$,  $\bm{h}_{\text{N}}$, and $\bar{P}$\\
 \STATE 1. {The \ac{BS} send training and control message to \ac{RIS}};\\
      \FOR{$i$ = 1:$N$}
        \STATE 2. The \ac{RIS} pick $i$-th beam $\bm{v}_i$ from the predefined codebook $\bm{V}$ and form the reflection matrix as $\Theta = \text{diag}\left(\bm{v}_i\right)$;\\
        \STATE 3. The \ac{IU} calculate the received power for the selected \ac{RIS} beam $\left\|\bm{h}_{\text{RI}}\text{diag}\left(\bm{v}_i\right)\bm{h}_{\text{BR}}\right\|^2$;\\
      \ENDFOR
\STATE 5. {The \ac{IU} feeds back the best beam index $i_\text{best}$ in terms of received power to the \ac{BS}};\\
\STATE 6. {The \ac{BS} obtains the precoder as  $\bm{w}_\text{DFT} = \frac{\bm{h}_\text{DFT}^{H}}{\left\|\bm{h}_\text{DFT}\right\|}$, where $\bm{h}_\text{DFT} = \bm{h}_{\text{RI}}\text{diag}\left(\bm{v}_{i_\text{best}}\right)\bm{h}_{\text{BR}}$ };\\
\STATE 7. {The \ac{RIS} generate
Compute power allocation $P_{\text{D}}$ and $P_{\text{R}}$ use methods in Sec. \ref{sec_power}, considering \ac{EMFE} constraint $\bar{P}$ under channel $\bm{h}_{\text{N}}$.}\\
\RETURN Optimal rate $R$ for \ac{IU} (\ref{eq_R_defi}).
 \end{algorithmic}
 \label{alg_2}
 \end{algorithm}

\subsection{Power Allocation Schemes for the Direct and RIS-assisted Links with EMFE Constraints}\label{sec_power}
Depending on the channel condition as well as the \ac{EMFE} constraint, the power division between the direct link $P_{\text{D}}$ and the \ac{RIS}-aided link $P_{\text{R}}$ may vary. In this subsection, various power allocation methods with different performance-complexity trade-offs are presented. Specifically, we consider the following:
\begin{itemize}
    \item \textbf{Method 1}: All power is allocated to the \ac{BS}-\ac{RIS}-\ac{IU} link, i.e., $P_{\text{R}} = P$.
    \item \textbf{Method 2}: All power is allocated to the direct link, i.e., $P_{\text{D}} = P$.
    \item \textbf{Method 3}: Based on the information from the \ac{NIU}, i.e., $P_{\text{N, Tx}}$ in (\ref{eq_power_niu}), the \ac{BS} fills the direct link with $P_{\text{D}} = P_{\text{N, Tx}}$ and allocate the remaining power to $P_{\text{R}} = P - P_{\text{D}}$.
    \item \textbf{Method 4}: As one upper bound of Method 3, the power allocation factor $\alpha$, i.e., 
    \begin{align}
        P_{\text{D}} = \alpha P, \alpha\leq\frac{P_{\text{N, Tx}}}{P},
    \end{align}
    can be adaptively and exhaustively optimized with the knowledge of the optimized beams from the corresponding steps in Algorithms \ref{alg_1}-\ref{alg_2}. 
    \item \textbf{Method 5}: As another upper bound of Method 3, we consider transmitting different codewords in the direct link and the \ac{RIS} link, i.e., the \ac{IU} rate can be upper bounded as
    \begin{align}\label{eq_R_inr}
    R_{\text{upper}} = & B\Bigg(\log_2\left(  1 +  \frac{P_\text{D}\left|\bm{h}_{\text{D}}\bm{w}_{\text{D}}\right|^2  }{BN_0}      \right) +\nonumber\\ & \log_2\left(  1 +  \frac{ P_\text{R}\left|\bm{h}_{\text{RI}}\bm{\Theta}\bm{h}_{\text{BR}}\bm{w}_{\text{R}}\right|^2 }{BN_0}      \right)\Bigg).
    \end{align}
    This can be achieved by, e.g., spatial multi-stream with joint detection at the \ac{IU}. In the following Lemma \ref{lemma1}, we show that $\alpha$ can be analytically determined when considering Method 5. 
\end{itemize}

\begin{lemma}\label{lemma1}
Define
\begin{align}
    c_1 = \frac{\mathbb{E}\left[\left|\bm{h}_{\text{RI}}\bm{\Theta}\bm{h}_{\text{BR}}\bm{w}_{\text{R}}\right|^2\right]}{BN_0},
\end{align}
and
\begin{align}
    c_2 = \frac{\mathbb{E}\left[\left|\bm{h}_{\text{D}}\bm{w}_{\text{D}}\right|^2\right]}{BN_0}
\end{align}
as the average received \ac{SNR} with optimized \ac{BS} and \ac{RIS} beamforming. The optimal power allocation can be calculated as
\begin{align}\label{eq_lemma1}
    \min\left\{\frac{P_{\text{N, Tx}}}{P}, \frac{Pc_2-Pc_1+P^2c_1c_2}{2\left(P^2c_1c_2\right)} \right\}.
\end{align}
\end{lemma}
\begin{proof}
The proof of Lemma \ref{lemma1} can be found in Appendix A.
\end{proof}

Lemma \ref{lemma1} is useful for evaluating the upper bounded performance in  \ac{RIS}-assisted networks with \ac{EMFE} constraint. With the knowledge of optimized beam pattern leading to average received \ac{SNR}, the optimal power allocation for the direct link and the \ac{RIS}-assisted link can be analytically derived.

\begin{lemma}\label{lemma2}
To obtain the best \ac{IU} rate (\ref{eq_R}), the considered \ac{EMFE}-constrained \ac{RIS} system should always fill in the direct link with the maximum possible power if the direct link is better in the sense that $c_2 > c_1$.
\end{lemma}
\begin{proof}
The proof of Lemma \ref{lemma2} can be found in Appendix B.
\end{proof}

%--------------------------------------------------
\section{Simulation Results} % 1.5 Page

In this section, we present numerical results for the proposed \ac{EMFE}-aware power allocation methods for \ac{RIS}-assisted networks. The simulation setup is presented in Fig. \ref{fig_3}. Consider one \ac{BS} with $N_{\text{T}} = 32$ antennas, one \ac{RIS} with $N = 100$ elements, one \ac{IU} and one \ac{NIU} both having one antenna. The carrier frequency $f_\text{c}$ is set to 28 GHz with 100 MHz channel bandwidth $B$. The noise power is set as -174 dBm/Hz with 10 dB noise figure. The antenna gain of the \ac{BS}, \ac{RIS}, and the users are set to 18 dBi, 18 dBi, and 0 dBi, respectively. The total transmit power is $P = $ 43 dBm in Fig. \ref{fig_r1}. As shown in Fig. \ref{fig_3}, the \ac{BS} is located at [-80m, 0] while the \ac{IU} is fixed at [80m, 0] for Case 1 and [-70m, 0] for Case 2. Except for otherwise stated, the \ac{RIS} is located at [0, 50m] for Case 1 and [-70m, 10m] for Case 2. Finally, for the \ac{NIU}, it moves along the x-axis for the Case 1 (Fig. \ref{fig_1}), while it stands close or far from the \ac{IU} (Case 2, Fig. \ref{fig_2}) with the same distance to the \ac{BS}. For simplification, we omit the \ac{EMFE} from the \ac{RIS}-\ac{NIU} link since the \ac{RIS} beam is supposed to be optimized towards \ac{IU}. Nevertheless, we verify this assumption in one of the curves in Fig. \ref{fig_r1}.

\begin{figure}
\centering
  \includegraphics[width=1.0\columnwidth]{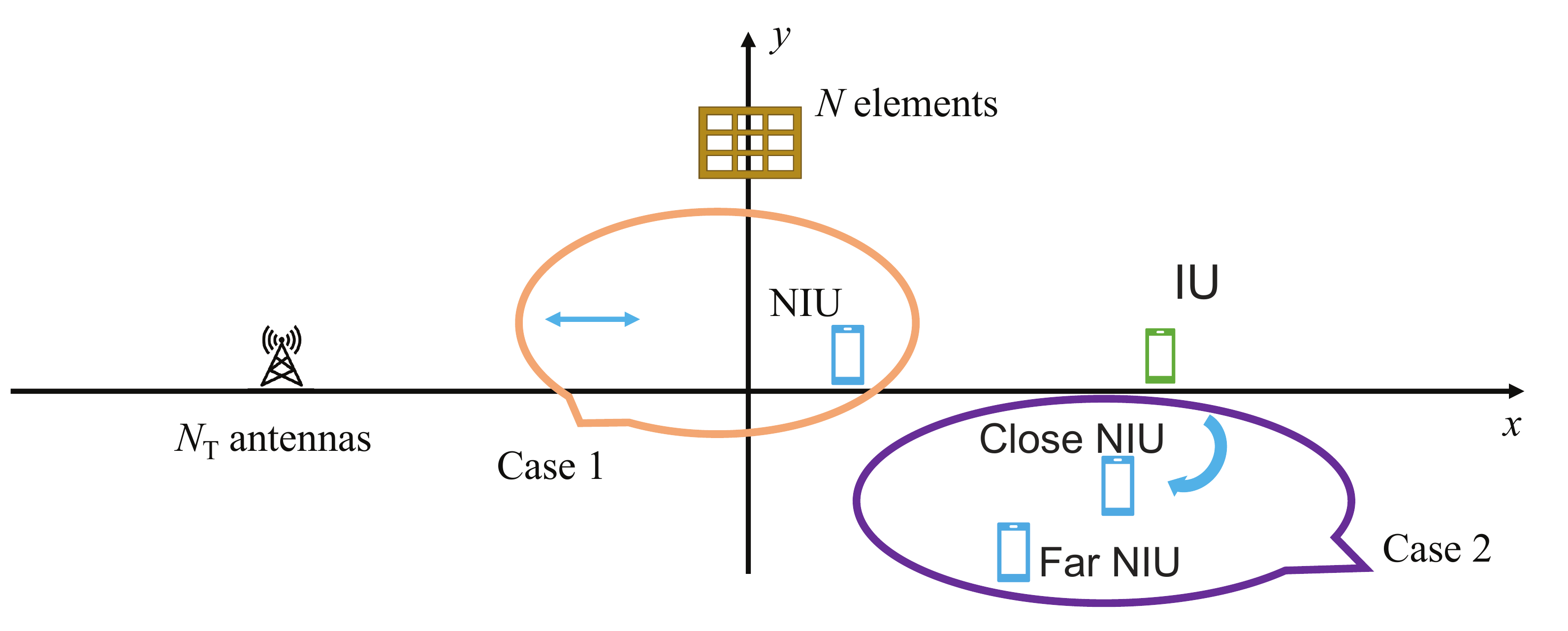}\\
\caption{The simulation setup for the two considered cases. }
\label{fig_3}
\end{figure}

Our proposed methods are generic to different channel models. We assume that we have all \ac{LoS} paths and the path loss at distance $d$ can be obtained by the mmMAGIC model \cite[Table IV]{theodore2017tap}
\begin{align}
    \text{PL} = 19.2\log_{10}(d) + 32.9 + 20.8\log_{10}(f_\text{c})
\end{align}
with 2 dB shadowing. Also, to fully evaluate the performance, for Case 1 (Fig. \ref{fig_1}) we consider Rayleigh fading with a unit variance while Case 2  (Fig. \ref{fig_2}) uses a more generic \ac{mmWave} channel with \ac{AoD} and multipath:
\begin{align}
    \bm{h} = \sqrt{\frac{\text{PL}}{L}}\sum_{l=1}^{L}\beta_l\bm{a}(\psi_l).
\end{align}
Here, $\beta\sim\mathcal{CN}(0,1)$ and  $L$ is the number of paths (set to 3 in the simulation). Furthermore, $\bm{a}(\psi_l) = \left\{e^{jkd_{a}\left(n_{\text{T}}-1\right)\sin\left(\psi_l\right)}\right\}_{n_{\text{T}} = 1}^{N_{\text{T}}}$ is the antenna steering vector with $k = 2\pi/\lambda$ and $d_{\text{a}} = \lambda/2$ when $\lambda$ is the wavelength. Moreover, the codebook-based beamforming proposed in Algorithm \ref{alg_2} can be applied with different codebooks. Here, we present results with \cite{Yu2018tsp}
\begin{align}
    \bm{V} = \left\{\bm{\mu}_i\right\}_{i = 1}^{\sqrt{N}} \otimes \left\{\bm{\nu}_j\right\}_{j = 1}^{\sqrt{N}},
\end{align}
where
\begin{align}
    \bm{\mu}_i = \left[1, e^{\frac{2\pi(i-1)}{\sqrt{N}}}, ...,  e^{\frac{2\pi(i-1)(\sqrt{N}-1)}{\sqrt{N}}}\right],
\end{align}
\begin{align}
    \bm{\nu}_j = \left[1, e^{\frac{2\pi(j-1)}{\sqrt{N}}}, ...,  e^{\frac{2\pi(j-1)(\sqrt{N}-1)}{\sqrt{N}}}\right],
\end{align}
and $\otimes$ represents Kronecker product.

\begin{figure}
\centering
  \includegraphics[width=1\columnwidth]{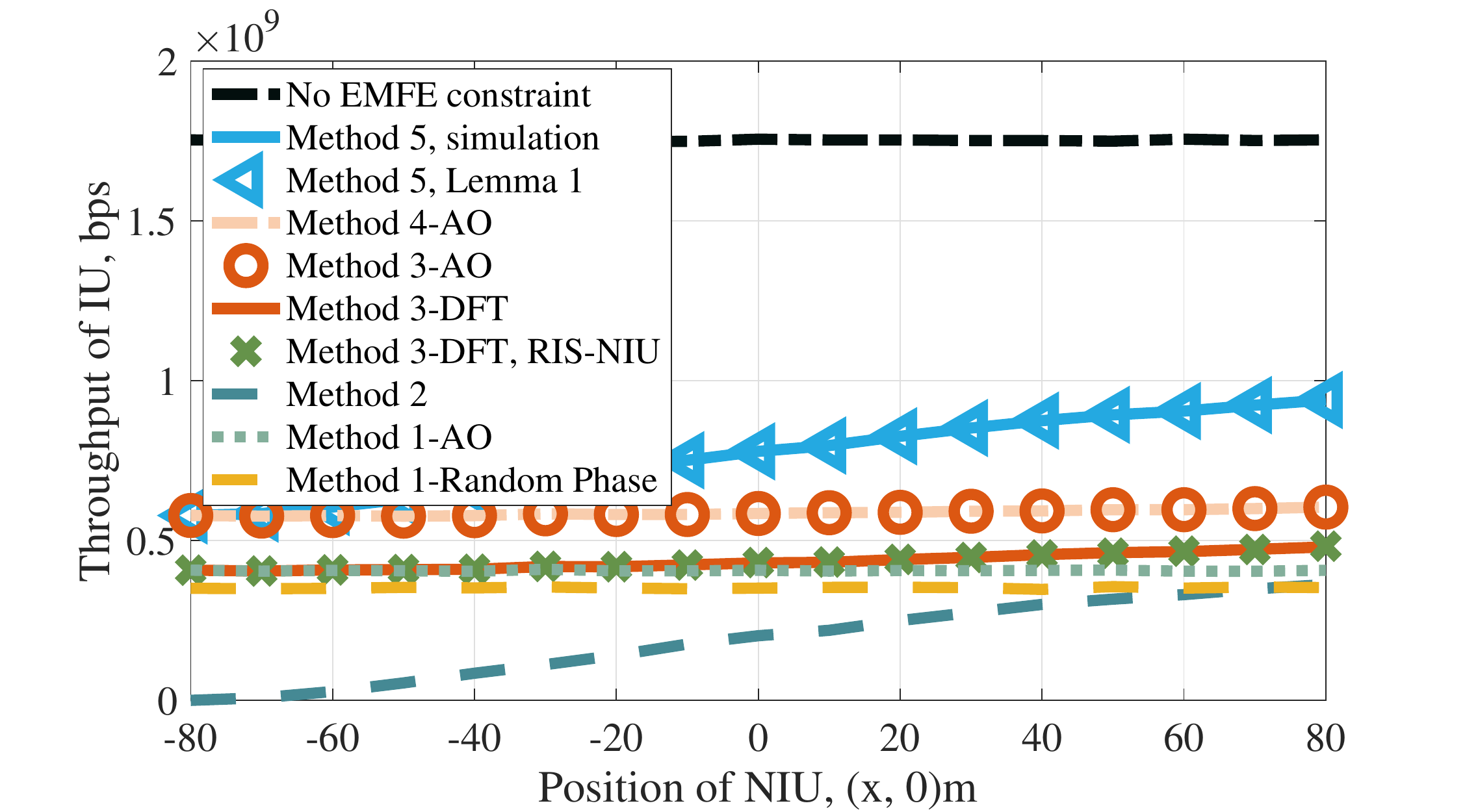}\\
\caption{\ac{IU} throughput as a function of the position of the \ac{NIU} for Case 1. The \ac{EMFE} constraint  $\bar{P}$ is set to 0.005 mW.}
\label{fig_r1}
\vspace{-5mm}
\end{figure}

Figure \ref{fig_r1} presents the \ac{IU} throughput as a function of the position of \ac{NIU} for Case 1, with the \ac{EMFE} constraint  $\bar{P}$ is set to 0.005 mW. Here, we consider the cases with Method 1 (\ac{RIS} only), Method 2 (direct link only), Method 3 with both \ac{AO} and \ac{DFT}-based beam optimization (corresponding to Algorithm \ref{alg_1} and \ref{alg_2}, respectively), Method 4 with exhaustive power allocation using \ac{AO}, Method 5 as an upper bound from both simulation and Lemma \ref{lemma1}, and the case without \ac{EMFE} constraint and transmit with the \ac{BS}-\ac{IU} direct link. The interference from \ac{RIS} to \ac{NIU} is considered in the case "Method 3-\ac{DFT}, \ac{RIS}-\ac{NIU}". As an additional benchmark, the performance with random phase \ac{RIS} is also presented. Then, in Fig. \ref{fig_r2}, with a more relaxed \ac{EMFE} constraint ($\bar{P} =$ 0.5 mW), we focus on the comparison of Methods 1-4 as well as  \ac{AO} and \ac{DFT}-based beam optimization.

To study the performance of Case 2 where the \ac{NIU} is alongside the \ac{IU}, in Fig. \ref{fig_r3}, we plot the \ac{IU} throughput versus the transmit power. Here, the \ac{EMFE} constraint  $\bar{P}$ set to 0.1 mW and the \ac{AoD} for the close and far \ac{NIU} are set to $\pi/16$ and $\pi/8$, respectively. We consider the cases with Method 1 (\ac{RIS} only), Method 2 (direct link only), and Method 3 with \ac{DFT}-based optimization (Algorithm \ref{alg_2}). Finally, to evaluate the impact of the deployment of the \ac{RIS}, both close and far \ac{RIS} with respect to the \ac{IU} are considered, and they have coordinates as (-70, 10)m, and (-30, 10)m, respectively.

According to these simulation results, the following conclusions can be drawn:
\begin{itemize}
    \item The \ac{EMFE} constraint from, e.g., regulation and standardization, could drastically affect the system performance, according to the no constraint case in Figs. \ref{fig_r1}-\ref{fig_r2} and saturation of Method 2 in Fig. \ref{fig_r3}, and thus needs to be carefully taking into account with network design.
    \item For the considered two cases, i.e.,  the \ac{NIU} between the \ac{BS} and the \ac{IU} (Case 1) and  alongside the \ac{IU} (Case 2), the system performance is affected by various parameters such as the strength of the \ac{EMFE} constraint, the position of \ac{RIS}, the total transmit power, as well as the relative distance between the \ac{NIU} and \ac{IU}.
    \item With the \ac{EMFE} constraint being considered, it is still beneficial to fully exploit the potential of the direct \ac{BS}-\ac{IU} path (which normally has better propagation conditions in terms of, e.g., path loss compared to the \ac{RIS} path), with the help of \ac{NIU}. As shown in a power-limited system in Figs. \ref{fig_r1}-\ref{fig_r2}, the performance of only using the direct path increases with smaller \ac{EMFE} constraints, i.e., further \ac{NIU}. On the other hand, with an increased transmit power budget in Fig. \ref{fig_r3}, \ac{RIS}-assisted link could eventually overperform even with a far-deployed \ac{RIS}.
    \item If the direct path has better propagation condition, it is preferable to fill the link first before considering the \ac{RIS}-assisted link (Lemma \ref{lemma2}), which is supported by the similar performance of Methods 3-4 in Figs. \ref{fig_r1}-\ref{fig_r2}.
    \item The analytical results of the  upper bound (\ref{eq_R_inr}) in Lemma \ref{eq_lemma1} agrees well with the simulations, and it reveals the system performance potential with transmitting different codewords for the direct and the \ac{RIS}-assisted link.
    \item As can be seen in Figs. \ref{fig_r1}-\ref{fig_r2}, the proposed \ac{DFT}-based scheme (Algorithm. \ref{alg_2}) could reach a close performance compared to the \ac{AO} method, especially with a more relaxed \ac{EMFE} constraint, and it does not require explicit \ac{CSI} when optimizing the beams. Also, random phase \ac{RIS} would result in a performance loss compared to \ac{AO}.
    \end{itemize}

\begin{figure}
\centering
  \includegraphics[width=1\columnwidth]{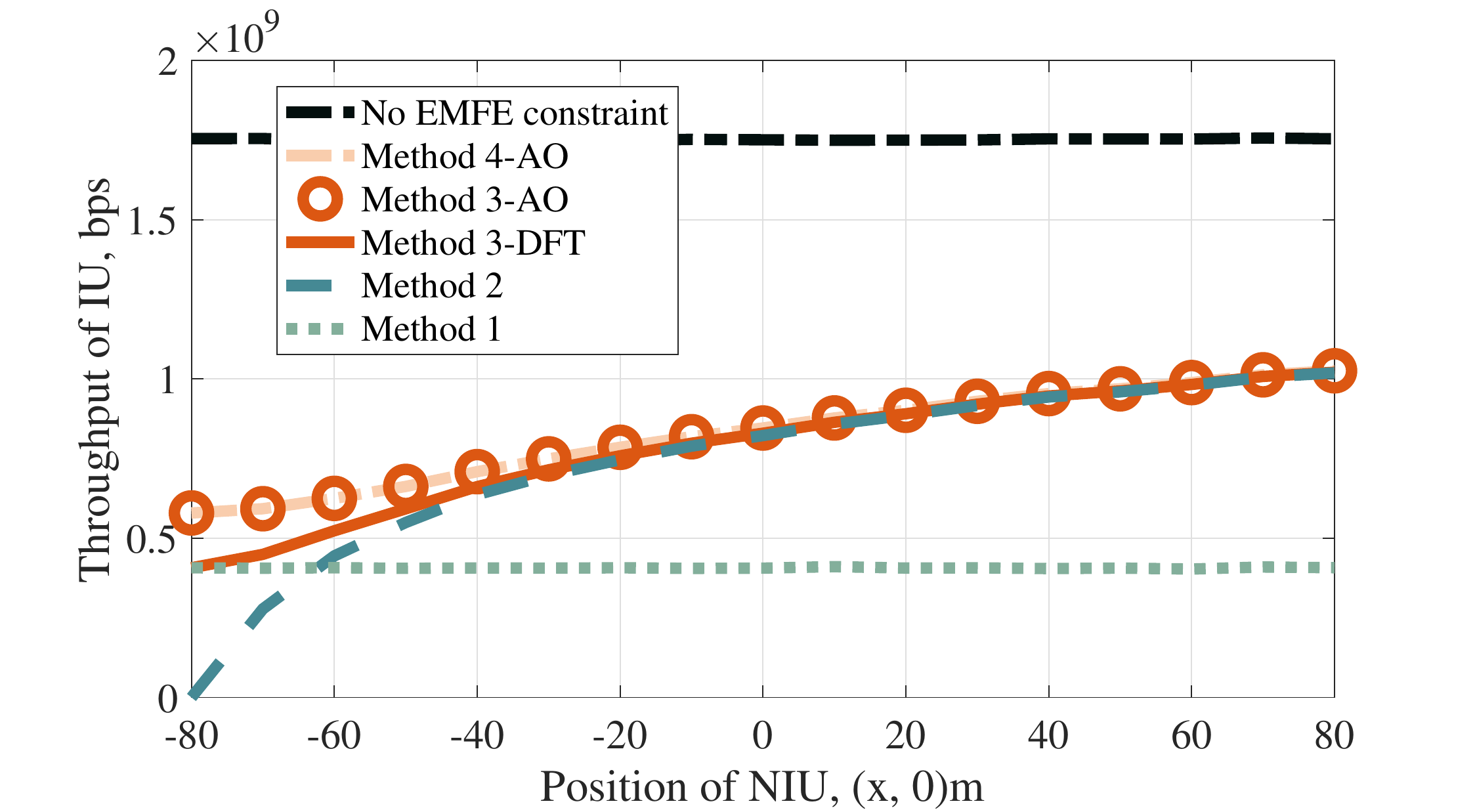}\\
\caption{\ac{IU} throughput as a function of the position of the \ac{NIU} for Case 1. The \ac{EMFE} constraint  $\bar{P}$ is set to 0.5 mW. }
\label{fig_r2}
\vspace{-3mm}
\end{figure}

\begin{figure}
\centering
  \includegraphics[width=1\columnwidth]{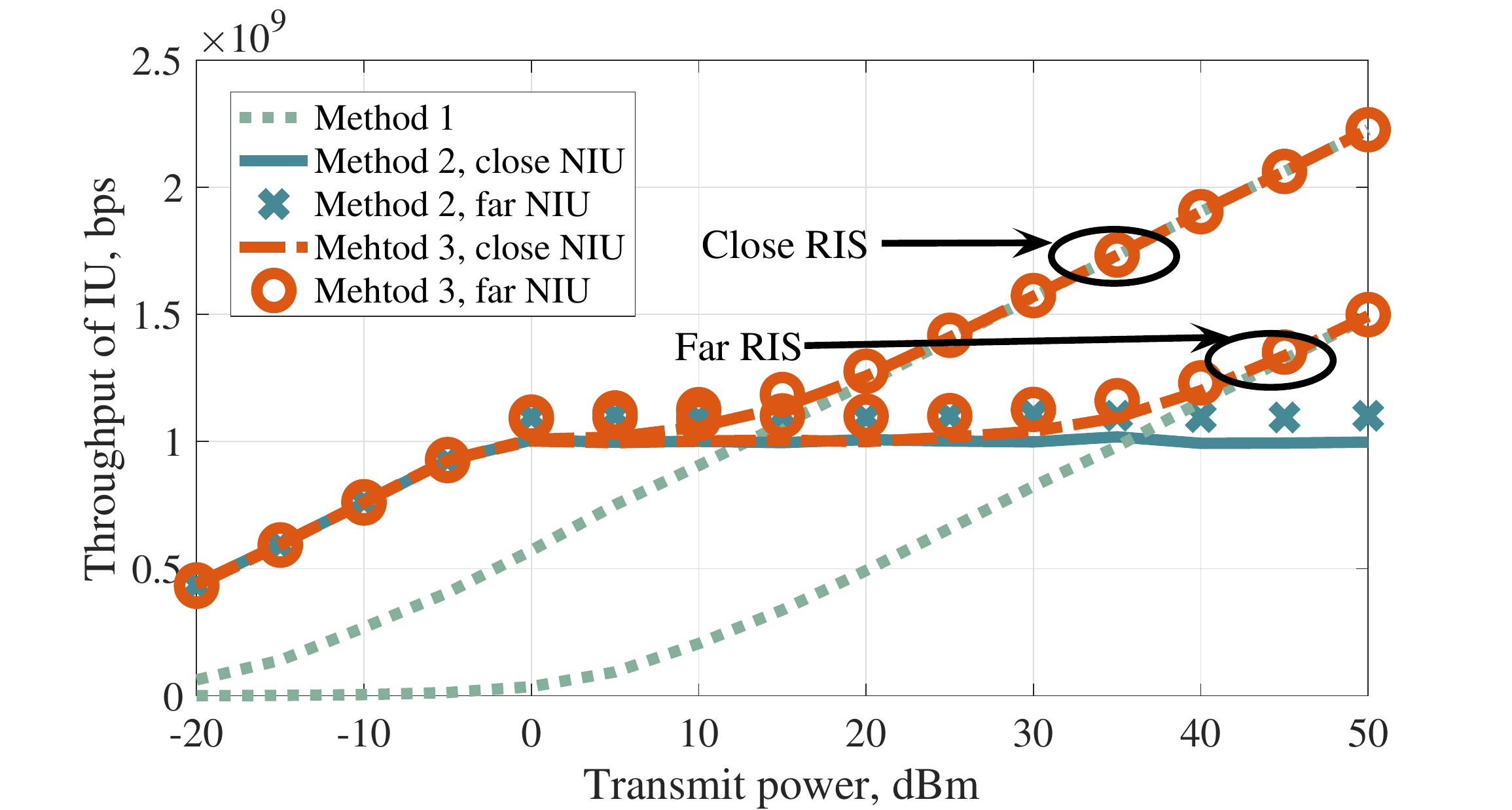}\\
\caption{\ac{IU} throughput as a function of transmit power $P$ for Case 2. The \ac{EMFE} constraint  $\bar{P}$ is set to 0.1 mW and the \ac{AoD} for the close and far \ac{NIU} are set to $\pi/16$ and $\pi/8$, respectively. The coordinates for the close and far \ac{RIS} are (-70, 10)m, and (-30, 10)m, respectively. }
\label{fig_r3}
\vspace{-6mm}
\end{figure}

%--------------------------------------------------
\section{Conclusions} % 1 Page with Reference
In this paper, we studied the \ac{EMFE}-limited \ac{RIS}-assisted system with two new types of constrained positions. Specifically, we reveal that with the help of \ac{NIUE} detection, it is beneficial in terms of the system throughput to fully utilize the better propagation path, which is normally the direct \ac{BS}-\ac{IU} link. Also, we evaluated the system performance with the proposed \ac{AO}- and \ac{DFT}-based scheme with various power allocation methods. The \ac{DFT}-based method is shown to be efficient with limited channel information. In addition, the analytical contributions on the upper bound of the system as well as the preference of the direction in power allocation provide solid insights on the considered setups.

%--------------------------------------------------
\section*{Acknowledgement}
This work  was supported by the EU H2020 RISE-6G project
under grant number 101017011.

\section*{Appendix A}\label{sec_appendixA}
With power allocation factor $\alpha$, the throughput of \ac{IU} (\ref{eq_R_inr}) can be written as
\begin{align}
    R_{\text{upper}} = B\left(\log_2\left(1+(1-\alpha)P c_1\right) + \log_2\left(1+\alpha P c_2\right)\right).
\end{align}
Then, the derivative of $ R_{\text{upper}}$ with respect to $\alpha$ is
\begin{align}\label{eq_deri}
  &  \frac{\text{d}R_{\text{upper}}}{\text{d}\alpha} =  B\log_2\left(\left(1+Pc_1-\alpha Pc_2\right)\left(1+\alpha Pc_2\right)\right)\nonumber\\&=  B\log_2\left( 1+Pc_1 + \alpha\left(Pc_2 +P^2c_1c_2 -Pc1 \right)- \alpha^2P^2c_1c_2\right).
\end{align}
Setting (\ref{eq_deri}) equal to zero the optimal power allocation for the upper bounded \ac{IU} rate (\ref{eq_R_inr}) can be obtained as \ref{eq_lemma1}.

\section*{Appendix B}\label{sec_appendixB}
(\ref{eq_R}) with power allocation factor $\alpha$ can be simplified as
\begin{align}
    R = B\log_2\left(1+\left(1-\alpha\right)Pc_1+\alpha Pc_2\right).
\end{align}
Its derivative with respect to $\alpha$ is
\begin{align}
     \frac{\text{d}R}{\text{d}\alpha} = \frac{BP\cdot\left(c_2-c_1\right)}{\left(Pc_2-Pc_1\right)x+Pc_1+1}.
\end{align}
When $c_2 > c_1$, i.e., the direct link is better, the derivative is always positive. As a result, the maximum allowed power ($P_{\text{N, Tx}}$) should be allocated to the direct path.

\bibliographystyle{IEEEtran}
%\vspace{-2mm}
\bibliography{main.bib}

\end{document}